\documentclass[runningheads,a4paper]{llncs} %

\usepackage{amsmath}
\usepackage{amssymb}

\usepackage{url}
\urldef{\mailsa}\path|{furer}@cse.psu.edu|    

\newcommand{\nn}{\mathbb{N}}
\newcommand{\op}{\mbox{op}}

\newcommand{\cc}{\mathbb{C}}

\begin{document}

\mainmatter  %

\title{How Fast Can We Multiply Large Integers 
on an Actual Computer?}

\titlerunning{How Fast Can We Multiply Large Integers?}

\author{Martin F\"urer%
\thanks{Research supported in part by NSF Grant CCF-0964655 and CCF-1320814.}}
\authorrunning{Martin F\"urer}

\institute{Department of Computer Science and Engineering \\
	Pennsylvania State University \\
	University Park, PA 16802,  USA \\
	furer@cse.psu.edu \\
\url{http://cse.psu.edu/~furer}}

\maketitle

\begin{abstract}
We provide two complexity measures that can be used to measure the running time of algorithms to compute multiplications of long integers. The random access machine with unit or logarithmic cost is not adequate for measuring the complexity of a task like multiplication of long integers. The Turing machine is more useful here, but fails to take into account the multiplication instruction for short integers, which is available on physical computing devices. 

An interesting outcome is that the proposed refined complexity measures do not rank the well known multiplication algorithms the same way as the Turing machine model. 
\end{abstract}

\section{Introduction}
The use of asymptotic time to measure the complexity of algorithms has been enormously successful in driving the search for new algorithmic ideas and more efficient algorithms. Nevertheless, it has some shortcomings.

One example is the complexity of the fastest known integer multiplication algorithm in the original version \cite{Furer2009} as well as the modular version \cite{DeKSS2008}. We call these algorithms F-R and DKSS-R respectively, as both operate over a ring.
Obviously, these algorithms running in time $n \log n \, 2^{O(\log^* \! n)}$ are asymptotically faster than the previously fastest algorithm \cite{SchonhageS1971} with a running time of $O(n \log n \log \log n)$.

Nevertheless, the Wikipedia entry ``Multiplication algorithm'' \cite{Wikipedia2013} says, ``However, these latter algorithms are only faster than Sch\"onhage-Strassen for impractically large inputs.'' And this judgment is not uncommon. It seems to be implied by the following argument. For any practical large length $n$ (from $2^{16} + 1$ to well beyond astronomical in length), $\log^* \! n$ is 5, resulting in $2^{\log^* \! n} = 32$, whereas for practical values of $n$, we have $\log \log n \leq 6$. (All logarithms are to the base 2 in this paper.)

In this special case, this reasoning is particularly faulty, because the exponent $\log^*n$ is used as an upper bound for the number of nested recursive calls, which should really be $\max\{0, \log^* \! n - 4\}$. The reason is that for small $n$, a different algorithm would be used in practice. Indeed, for any practical length $n$, the number of nested recursive calls and thus the exponent should be at most 1. 
The practical performance of the newer algorithm would have to be determined by an implementation, which might well be competitive, even though not by a large factor.

Thus we have these two very different methods of evaluating a multiplication algorithm, the theoretical asymptotic time bound and the practical implementation. The purpose of this paper is to build a bridge between the two, i.e., to propose models of computation, that are theoretically rigorous, yet can better predict the practical running time. 

We have to stress, that we are focussing here on unbounded integer operations. The models developed here are valid for similar tasks. This is in contrast to many areas, like graph algorithms, where the unit cost RAM provides a perfectly good complexity measure, because most natural algorithms involve only numbers of length $O(\log n)$ which can be implemented to fit into a computer word.

Another example shows our concern more clearly.
Sch\"onhage and Strassen~\cite{SchonhageS1971} have designed two fast integer multiplication algorithms. The first one (SS-C) is based on numerical approximations of complex roots of unity. It runs in time $O(n \log^2 n)$ for integers of length $n$ if school multiplication is used for all recursive calls. The second one (SS-F) is a discrete algorithm based on integers modulo Fermat numbers. It runs in time $O(n \log n \log \log n)$. Yet for a long time, the first algorithm has been routinely used for the extensive integer multiplications needed to find large prime numbers.

The first algorithm, SS-C, is asymptotically slower and somewhat unnatural for a discrete problem. It requires the tedious task of controlling rounding errors. It seems justified to ask whether the greater simplicity of SS-C is a sufficient reason to select it. After all, implementers don't stick with the simplicity, but add many clever ideas to speed up an implementation.

The running times cited above are bounds that hold simultaneously for the Turing machine time, as well as for the Boolean circuit size. These computation models and their time complexity are very natural and have facilitated many new algorithmic ideas in a vast number of areas including fast integer multiplication.

Sometimes, Turing machines are viewed as being impractical, because they function quite differently from physical computing machines. But when there is strong locality of memory access and no need for random access, as is the case for the Fast Fourier Transforms (FFTs) used in these fast integer multiplication algorithms, a Turing machine actually works just fine. Indeed, the first implementation \cite{SchonhageGV94} of a fast integer multiplication algorithm was on a Turing machine (a versatile multi-tape Turing machine with alphabet size equal to $2^w$, where $w$ is the word-length of the machine).

So in this particular application, the disadvantage of the Turing machine is not the lack of random access. What is missing in the Turing machine model are the built-in arithmetic operations of an actual computing device, in particular the multiplication instruction. Thus, we are aiming at a version of a Random Access Machine (RAM) \cite{ShepherdsonS63,CookR1973} with multiplication.

The unit cost RAM with multiplication is not a viable model, because it can quickly build up large numbers, and its operations are then far too powerful compared to a real world computing device. We want a version of a RAM that is both theoretically appealing and closely modeling the capabilities of actual computers.

The basic model proposed here is a $\log$-RAM. It can do arithmetic operations of length $O(\log n)$ in constant time. While this model might be too restricted for very short inputs $n$, it is very realistic for input lengths from thousands to trillions and beyond.

A more refined model aims to be even more realistic, adjusting for the varying costs of different operations, and accounting for the benefits of temporal and spatial locality.

\section{The Basic Model}

Our $\log$-RAM model is a random access machine \cite{ShepherdsonS63,CookR1973} augmented with arithmetic, Boolean, shift, input and output operations on nonnegative integers stored in binary. Furthermore, it has conditional and unconditional jump operations. It can do direct and indirect addressing. It has register $R_i$ for every $i \in \nn$, as well as an input and output register.

To be specific, we give now a precise definition of the instruction set of a $\log$-RAM. Nevertheless, the important part of a $\log$-RAM is not the instruction set, but the time complexity.

\begin{definition}
 A \emph{$\log$-RAM} has the the following instructions set
\begin{itemize}
\item 
$R_i = A \; \op \; A'$. \\
This is an assignment. Here, and in the following, the arguments $A$ and $A'$ can be of one of 3 possible forms: nonnegative integer constant, $R_j$, or $R_{R_j}$. 
The operation is one of $\{\lor, \land, \oplus, \lnot, +,-,*,\div\}$, where the Boolean operations (or, and, sum modulo 2, and negation) are vector operations, and $div$ is integer division ($A \div A' = \lfloor A/A' \rfloor$). 
The special case $R_i = R_j + 0$ copies registers.

Registers and arguments are non-negative integers stored in binary. For Boolean operations, they are interpreted as bit vectors.
\item
$R_i = R_j \; \mbox{cyclic shift}  \pm A \; \mbox{of length} \; A'$. \\
Here the  $A'$ right-most bits are cyclicly shifted by $A$ positions to the left ($+A$) or to the right ($-A$).

\item
Input to $R_i$ from $A$ to $A'$. \\
This instruction has the effect of reading into $R_i$ the $A'-A+1$ bits from position $A$ to position $A'$ of the input register.
\item
Output $R_j$ of length $A$. \\
This instruction has the effect of concatenating the rightmost $A$ bits of $R_j$ to the output on the right hand side.
\item
Jump if $A=A'$. \\
This instruction allows a jump conditioned on two registers being equal, but also a jump when 0 ($R_i = 0$) and an unconditional jump ($0=0$).
\end{itemize}
\end{definition}

\begin{definition}
We assume the machine knows the length $n$ of the input, e.g., it is written in register $R_0$ before a computation starts. \\
All registers, except the input and output register, are only allowed to be assigned bit strings of length $O(\log n)$ encoding nonnegative integers up to $n^{O(1)}$. Arguments denoting positions and lengths in the input, output, and shift instructions are only allowed to have values $O(\log n)$. \\
The time for any ($O(\log n)$ long) operation is $O(1)$ including for multiplication and division. We therefore might refer to the machine model as a \emph{unit cost $\log$-RAM}.
\end{definition}

As the log-RAM realistically mimics physical computing machines, it allows for a speed-up of $\Theta(\log n)$ for additions and shifts  of length $\Theta(\log n)$ and a speed-up of $\Theta(M(\log n))$ for multiplications of binary integers compared to Turing machine time, where $M(n)$ is the multiplication time of a Turing machine.

An alternative definition, equivalent for our purposes, would be to allow arbitrary long registers and arguments, but to charge $(\lceil \ell / \log n \rceil)^2$ for multiplication and division instructions, 1 for jump instructions, and $\lceil \ell / \log n \rceil$ for all other instructions. Here, $\ell = \max\{\ell',2^{\ell''}\}$, where $\ell'$ is the maximal length of any operand, and $\ell''$ the maximal length of a position or length operand.

Naturally, the equivalence only holds for somewhat efficient computations. The alternative definition of the $\log$-RAM would define a universal computing device, while by our definition, a $\log$-RAM can only define polynomial space functions.

\begin{proposition}
 A partial function is computable by a $\log$-RAM if and only if it is computable by a Turing machine in polynomial space.
\end{proposition}

\begin{proof}
 Only if part: Note that the the $\log$-RAM can only access the first $n^{O(1)}$ registers, which each are allowed to hold bit strings of length $O(\log n)$. The $\log$-RAM has no way to address other registers. Each operation can easily be simulated by a Turing machine.
 
 If part: The $\log$-RAM can simulate a polynomially space bounded Turing machine by storing the contents of tape cell $i$ in the register $R_{i+2}$ and storing the head position in register $R_2$.
\qed
\end{proof}

\section{Differences to the Traditional RAM}
Naturally, the definition of the log-RAM is very similar to the traditional definition of  a RAM \cite{ShepherdsonS63,CookR1973} (see \cite [pp. 5 ff.]{AhoHU74}). Nevertheless, there are important differences.

The unit cost RAM provides an excellent cost measure for well behaved algorithms, e.g., for many graph algorithms. It gets useless if a multiplication operation is allowed, as one could quickly produce integers of huge lengths, resulting (for some simple instruction sets) in the power of unbounded parallel machines, which can handle PSPACE in polynomial time \cite{HartmanisS74}. 
Even without multiplication, length $T(n)$ integers can be produced in time $T(n)$, resulting in unrealistically cheap additions.

These drastic problems are avoided by the traditional RAM with logarithmic cost, which is a cost proportional to the length of a binary integer. Nevertheless, this type of RAM is still not suitable to provide a practical cost measure for a task like the multiplication of long integers. Obviously an instruction multiplying in one step with a cost of $O(n)$ would make the task trivial. Without a multiplication instruction, the logarithmic cost RAM still has the random access advantage over Turing machines (which we don't need here). But it does not have the practical advantage of real computers, that can do operations like additions of reasonable numbers almost as fast as a bit operation.

Our new log-RAM provides a complexity measure that is much closer to the computation time of a real computer. It allows a theoretical investigation that can better predict the practicality of an algorithm in a domain like large integer multiplication. 

One could even define a log-RAM with a more explicit cost function. If all registers have lengths bounded by $k \log n$, then the cost of a multiplication could be defined as $k^2$ and the cost of any other operation could be defined as $k$. This time measure would avoid any large hidden constant factors.

\section{Performance of the log-RAM on Multiplication Algorithms}

First let us review the most important multiplication algorithms. 

\subsection{The traditional multiplication algorithms}
The first multiplication algorithm with a non-trivial asymptotic running time is due to Karatsuba \cite{KaratsubaO1962}. It multiplies $a_1 2^{n/2} + a_0$ with 
$b_1 2^{n/2} + b_0$ recursively by computing the 3 products $a_1 b_1$, $a_0 b_0$, and $(a_1 + a_0)(b_1 + b_0)$, to obtain the product with a few additions and subtractions. The running time is $O(n^{\log 3})$. 

It is straightforward to see that working with numbers of length $O(\log n)$ allows us to use the full computational power of the $\log$-RAM. Thus, the time for school multiplication is $O(n^2 / \log^2 n)$, while the time for Karatsuba's algorithm is $O(n^{\log 3} / \log^2 n)$. Karatsuba's algorithm can be viewed as multiplying 2 linear polynomials by evaluating them at $0$, $1$, and $\infty$, followed by multiplying the values and interpolating. Toom's algorithm \cite{Toom1963} (analyzed and implemented on the Turing machine by Cook \cite{Cook1966}) instead uses higher degree polynomials. Degree 2 (with 5 coefficients in the product polynomial) is often used for moderately large numbers. Clearly, on a $\log$-RAM, we get the same factor $\Theta(\log^2 n)$ speed-up.

The first Sch\"onhage-Strassen integer multiplication algorithm, SS-C \cite{SchonhageS1971} partitions the factors into pieces of length $\Theta(\log n)$, to be used as coefficients of two polynomials over the complex numbers $\cc$. The polynomials are evaluated at all powers of a primitive root of unity by a Fast Fourier Transform (FFT). After the multiplications of corresponding values, interpolation is done by an inverse FFT. $O(\log n)$ accuracy of these numerical computations is sufficient to recover the precise result by rounding.
The running time is $O(n \log^2 n)$ on the Turing machine.

The second Sch\"onhage-Strassen integer multiplication algorithm, SS-F \cite{SchonhageS1971} partitions the factors into pieces of length $\Theta(\sqrt{n})$, to be used as coefficients of two polynomials over the ring of integers modulo $2^{\sqrt{n}}+1$, where $\sqrt{n}$ is rounded to a power of 2. The polynomials are evaluated at all powers of a principal root of unity in this ring by an FFT. The multiplications of values is done recursively. It is followed by interpolation with an inverse FFT. This faster method requires a depth $O(\log \log n)$ of nested recursions, and runs in time $O(n \log n \log \log n)$ on the Turing machine.

It is important to notice that the slower SS-C algorithm does small multiplications on every level of the FFT. Thus it can really profit from a built-in multiplication instruction for short integers. The faster SS-F algorithm does relatively simple shift operations at all levels of the FFT. This makes it fast for Turing machines. On the other hand, it cannot benefit from a built-in integer multiplication for small integers, except at the bottom of the recursion.
 
\begin{theorem}
 (a) The running time of the first Sch\"onhage-Strassen integer multiplication algorithm, SS-C, is $O(n)$ on the $\log$-RAM. \\
  (b) The running time of the second Sch\"onhage-Strassen integer multiplication algorithm, SS-F, is $O(n \log \log n)$ on the $\log$-RAM. \\
\end{theorem}
 
\begin{proof}
(a) The analysis of the Turing machine algorithm accounts for $O(n \log n)$ simple bit operations to do shifts, copies and additions for the FFT and its inverse. Furthermore, it accounts for $O(n/\log n)$ multiplications of length $O(\log n)$ at each of the $O(\log n)$ levels of the FFT. As all coefficients have lengths $\Omega(\log n)$, all simple operations, including input and output, allow for a speed-up by a factor of $\Theta(\log n)$ on the $\log$-RAM compared to the Turing machine. Furthermore, all the $O(n)$ small integer multiplications are done in constant time each on the $\log$-RAM, resulting in an overall linear time.

(b) When coefficients reach a length of $O(\log n)$, recursion is no longer required. Then multiplications can be done directly on the $\log$-RAM. This reduces the recursion depth of the FFT from $\log \log n$ to $\log \log n - \log \log \log n$, representing no asymptotic speed-up. The Turing machine algorithm spends time $O(n \log n)$ at each recursion level. With the shortcut just described, all numbers involved have lengths $\Omega(\log n)$, resulting in a speed-up by a factor of $\Theta(\log n)$ on the $\log$-RAM for the additions and shifts. Therefore, the time for all additions and shifts is $O(n \log \log n)$.
There are $O(n 2^{\log \log n - \log \log \log n} / \log n) = O(n / \log \log n)$ short multiplications. Thus the total cost of $O(n \log \log n)$ is determined by the additions and shifts.
\qed
\end{proof}

\subsection{The newest multiplication algorithms}
 
The asymptotically fastest multiplication algorithm F-R \cite{Furer2009} does the FFT over a ring of polynomials $\mathcal{R} = \cc[x]/ (x^P + 1)$ with both, the value of $P$ and the length of the coefficients being of order $\log n$. Of the $O(\log n)$ levels of the FFT, only every $\log \log n$-th level requires expensive multiplications in the ring $\mathcal{R}$, while the multiplications at the other levels are done by a version of cyclic shifts.

Multiplication in the ring $\mathcal{R}$ itself is done by an FFT with $O(\log \log n)$ levels of cheap multiplications by cyclic shifts. For the asymptotic analysis of the Turing machine algorithm, the multiplication of values is done recursively. This introduces the factor of $2^{O(\log^* \! n)}$ for the $\log^* \! n$ depth of recursive calls with geometrically increasing cost from one recursion depth to the next. A practical implementation would not do such recursive calls, but instead use the built-in multiplication instruction of the actual machine. Similarly, the $\log$-RAM does each such multiplication with $1$ machine operation (or $O(1)$ operations with a small hidden factor).

\begin{theorem}
 The running time of the F-R integer multiplication algorithm is $O(n)$ on the $\log$-RAM.
\end{theorem}
 
\begin{proof}
A multiplication in the ring $\mathcal{R}$ requires $O(\log \log n)$ levels of $O(\log n)$ easy operations of length $O(\log n)$ and direct multiplication of $O(\log n)$ values of length $O(\log n)$ each. The resulting cost for one multiplication in $\mathcal{R}$ is $O(\log n \log \log n)$ for the easy operations and $O(\log n)$ for the direct multiplications. Thus the total time is $O(\log n \log \log n)$.

The FFT and its inverse have $O(\log n)$ levels involving shifts and additions of integers of length $\Omega(\log n)$ and a total length of $O(n)$ on each level. This requires time $O(n)$ on the $\log$-RAM. 

Furthermore, there are also $O(\log n / \log \log n)$ levels, with each requiring $O(n/ \log^2 n)$ multiplications in $\mathcal{R}$ at a cost of $O(\log n \log \log n)$ each. Thus this part also requires time $O(n)$.
\qed
\end{proof}
 
The discrete variant DKSS-R \cite{DeKSS2008} of the F-R algorithm operating over a ring approximating $p$-adic numbers has the same time analysis for the $\log$-RAM as the original F-R algorithm.
 
\subsection{Comparisons of the $\log$-RAM algorithms}
 
In summary, the algorithms SS-C, SS-F, F-R, and DKSS-R still have very similar running times in the $\log$-RAM model. It is interesting to see that the ranking changes. The previously slowest (SS-C) and the most advanced (F-R and DKSS-R) now have the same linear running time. The elegant SS-F algorithm, that achieves multiplication mainly by shifts is now less competitive, because it does not make much use of the power of the built-in multiplication instruction available on the $\log$-RAM and in practical computers.

Still it is hard to judge these algorithms just based on the $\log$-RAM complexity. The algorithms SS-C and F-R have the significant drawback of using non-discrete numerical computations requiring sufficient precision and rounding.

\section{Related Tasks on the $\log$-RAM and the Storage Modification Machine}
\subsubsection{Division and elementary functions}
The results about multiplication have immediate corollaries about division and other tasks based on multiplication. Division and $n$-bit approximations of algebraic numbers can be computed with the help of these algorithms with the same asymptotic running time, while $n$-bit approximate evaluations of elementary functions or constants like $e$ and $\pi$ can be computed with only an additional factor of $\log n$ time \cite{Brent1976}.

\subsubsection{Storage modification machine}

It is worth noticing that there is another theoretically interesting modification of the random access machine (RAM), namely the storage modification machine. Here the modification goes in the opposite direction. Instead of allowing the more powerful multiplication instruction, even the weaker addition instruction is not allowed. Instead, there is just a successor and a predecessor instruction. As numbers never get too long, the unit cost measure makes sense here. 

Sch\"onhage \cite{Schonhage1980} has obtained a very powerful result about multiplication on a storage modification machine. It can be done in linear time. Naturally, this is a far more difficult result than our linear time results on the $\log$-RAM. On the other hand, his algorithm is very sophisticated with sorting and table look-up for the mass production of short products. This is seemingly an impractical algorithm. In the linear time $\log$-RAM algorithms, on the other hand, we can use natural practical procedures and take advantage of easily available hardware.

\subsubsection{Addition versus multiplication}
An added bonus of the $\log$-RAM model is also the distinction between the easy task of addition and the complicated task of multiplication. On the storage modification machine both tasks take linear time, just with a huge difference of the constant factors involved. The $\log$-RAM model makes the distinction clear. Using the built-in addition instruction, long additions of length $n$ take time $O(n/\log n)$, which is significantly better than the $O(n)$ time for multiplication.

\subsubsection{Open Problem}
For the Turing machine model, there is no super-linear lower bound known for  integer multiplication. For the $\log$-RAM model the input only provides a trivial $\Omega(n/\log n)$ lower bound. Naturally, we conjecture an $\Omega(n)$ lower bound for the $\log$-RAM. 

This conjectured $\Omega(n)$ lower bound for multiplication on the $\log$-RAM might be independent of the well known $\Omega(n \log n)$ lower bound conjectured for multiplication on the Turing machine. The $\log$-RAM can only benefit strongly from a Turing machine computation if it operates on $\Theta(\log n)$ long chunks of data, while the Turing machine can simulate a $\log$-RAM
efficiently only if the $\log$-RAM does not jump around much.

\section{The Refined Model log-RAM with Depth-Cost}
 
The proposed $\log$-RAM model realistically models the advantage provided by the availability of instructions (including multiplication) operating on words in real machines. Still, it does not account for the higher cost of a multiplication instruction over a simple Boolean vector operation on a word.
More importantly, the $\log$-RAM model does not model the speed-up provided by 
local access patterns on real machines.

For this purpose, we propose the \emph{$\log$-RAM with Depth-Cost} as a refined model. We basically keep the instruction set of the $\log$-RAM, but define a different cost measure. The cost is modified only slightly, as the $\log$-RAM already quite well approximates the cost on an actual machine.

\begin{definition}
A $\log$-RAM with Depth-Cost has all the instructions of the (unit-cost) $\log$-RAM plus an additional vector copying instruction. 
The vector copying instruction puts a copy of the vector of the memory cells from $R_A$ to $R_{A+A''-1}$ into the memory cells from $R_{A'}$ to $R_{A'+A''-1}$. The cost of this operation is $\max\{\log A, \log A', A''\}$.
For the other operations, there is a cost associated with accessing a register and cost associated with the operation itself. \\
(a) The cost of an operation is the order of the parallel time to do this operation efficiently on operands of length $O(\log n)$ by a bounded fan-in Boolean circuit. In particular, the cost of Boolean operations is $O(1)$, while the cost of arithmetic operations and shifts is $O(\log \log n)$. \\
(b) The cost of accessing register $R_i$ is $O(\max\{1, \log i\})$.
\end{definition}

The cost of accessing $R_i$ reflects the idea that registers with low index $i$ are in a faster cache and therefore less costly to access. 
The cost of the vector copying operation reflects locality of access. Access in a single memory cell is expensive. Accessing a sequence of adjacent memory cells with the vector copying instructions is cheaper.

Due to the locality of access for doing FFTs, the cost of memory access can be bounded by the cost of arithmetic operations, if one uses a cache hierarchy, with geometrically increasing chunks of data being brought in or moved out. A higher cache level just means shorter addresses, i.e., being closer to the top of the memory.

\section{Multiplication on the log-RAM with Depth-Cost}
Trivial tasks, like input, output and addition of length $n$ numbers now take $O(n/\log n)$ steps costing $O(\log \log n)$ each, resulting in time $O(n \log \log n / \log n)$ per task.

The simple multiplication algorithms (school, Karatsuba, Toom) gain a factor of $\Omega(\log n /\log \log n)$ compared to Turing machine cost, because they operate with chunks of length $\Theta(\log n)$ stored in a register costing $O(\log \log n)$ per operation.

Accessing far away registers is efficient for these and all the other multiplication algorithms studied here. They would just be brought to the front in vectors of $O(\log n)$ length at a cost of $O(1)$ per register.

\begin{theorem}
The first Sch\"onhage-Strassen SS-C algorithm has a running time of 
$O(n \log \log n)$ on the $\log$-RAM with Depth-Cost.
\end{theorem}

\begin{proof}
All operations now cost $O(\log \log n)$ instead of $O(1)$ in the unit cost $\log$-RAM model. 
\qed
\end{proof}

\begin{theorem}
The second Sch\"onhage-Strassen SS-F algorithm has a running time of 
$O(n (\log \log n)^2)$ on the $\log$-RAM with Depth-Cost.
\end{theorem}

\begin{proof}
A direct implementation (still with bringing vectors of length $O(\log n)$ to the front at once) costs an additional factor of $O(\log \log n)$ compared to the unit cost $\log$-RAM model because of the cyclic shifts and additions. 
\qed
\end{proof}

The cost of $O(\log \log n)$ for an addition could be reduced to $O(1)$ by storing numbers in a redundant form.
Every number would be represented as the sum of $2$ registers. An actual sum would require the replacement of $4$ summands by $2$, which can be accomplished by $O(1)$ Boolean and shift operations, as there are no long carries to handle. But these cost savings are useless, because the shift operations still cost $O(\log \log n)$ each.

\begin{theorem}
  The F-R algorithm has a running time of 
$O(n \log \log n)$ on the $\log$-RAM with Depth-Cost.
\end{theorem}

\begin{proof}
 In the F-R algorithm, most levels of the FFT do cheap operations with cyclic shifts of coefficients within the ring $\mathcal{R}$. The coefficients themselves are not subjected to shifts, they are just cyclicly interchanged. These shifts can be done by $O(1)$ vector copy operations at a cost of $O(\log n)$ per operation involving a vector of length $O(\log n)$. The coefficients are subject to the addition operation. But this time, as there are no cyclic shifts of the summands, redundant additions are sufficient at a cost of $O(1)$ per operation.
 
 Only every $O(\log \log n)$-th level, expensive operations (arbitrary multiplications in $\mathcal{R}$) have to be done. In each of the $O(\log n / \log \log n)$ expensive levels, $O(n / \log^2 n)$ multiplications in $\mathcal{R}$ have to be done at a cost of $O(\log n (\log \log n)^2)$ each. This results in a total cost of $O(n \log \log n)$.
\qed
\end{proof}

\section{Conclusions}
We have investigated the cost of integer multiplication and noticed that neither the Turing machine nor the standard random access machine (RAM) models provide good complexity measures. The Turing machine, as well as the RAM without multiplication instruction cannot model the advantage of  physical machines with a multiplication instruction on words. 

If the unit cost RAM had a multiplication instruction, then it would not be a good computation model, because it has the power of parallel machines by creating huge numbers. The logarithmic cost RAM could have a multiplication instruction. But it would show too low a cost for long integers and too high a cost for short integers. It would not be a reasonable model for measuring the complexity of long integer multiplication.

We have proposed two RAM variants as better complexity measures for long multiplication and similar tasks. The measures reflect the fact that addition is faster than multiplication by known algorithms. They provide interesting practical complexity results for the various known algorithms. It may be surprising that the order of the standard multiplication algorithms by their time complexity in the new measures is different from the corresponding order in the Turing machine measure. The often used SS-C algorithm, based on numerical approximation in $\cc$, is actually faster in the new measures than the discrete SS-F algorithm, which is faster in the Turing machine model.


\begin{thebibliography}{10}

\bibitem{AhoHU74}
A.~V. Aho, J.~E. Hopcroft, and J.~D. Ullman.
\newblock {\em The Design and Analysis of Computer Algorithms}.
\newblock Addison-Wesley, Reading, Mass, 1974.

\bibitem{Brent1976}
R.~P. Brent.
\newblock Fast multiple-precision evaluation of elementary functions.
\newblock {\em J. Assoc. Comput. Mach.}, 23:242--251, 1976.

\bibitem{Cook1966}
S.~A. Cook.
\newblock {\em On the minimum computation time of functions}.
\newblock PhD thesis, Harvard University, 1966.

\bibitem{CookR1973}
S.~A. Cook and R.~A. Reckhow.
\newblock Time bounded random access machines.
\newblock {\em Journal of Computer and System Sciences}, 7(4):354--375, Aug.
  1973.

\bibitem{DeKSS2008}
A.~De, P.~Kurur, C.~Saha, and R.~Saptharishi.
\newblock Fast integer multiplication using modular arithmetic.
\newblock In {\em STOC '08: Proceedings of the 40th annual ACM symposium on
  Theory of computing}, pages 499--506, New York, NY, USA, 2008. ACM.

\bibitem{Furer2009}
M.~F{\"u}rer.
\newblock Faster integer multiplication.
\newblock {\em SIAM Journal on Computing}, 39(3):979--1005, 2009.

\bibitem{HartmanisS74}
J.~Hartmanis and J.~Simon.
\newblock On the power of multiplication in random access machines.
\newblock {\em Foundations of Computer Science, IEEE Annual Symposium on},
  0:13--23, 1974.

\bibitem{KaratsubaO1962}
A.~Karatsuba and Y.~Ofman.
\newblock Multiplication of multidigit numbers on automata.
\newblock {\em Doklady Akademii Nauk SSSR}, 145(2):293--294, 1962.
\newblock (in Russian). English translation in Soviet Physics-Doklady 7,
  595-596,1963.

\bibitem{Schonhage1980}
A.~Sch{\"o}nhage.
\newblock Storage modification machines.
\newblock {\em SIAM J. Comput.}, 9(3):490--508, August 1980.

\bibitem{SchonhageGV94}
A.~Sch{\"o}nhage, A.~F.~W. Grotefeld, and E.~Vetter.
\newblock {\em Fast algorithms: {A} {T}uring machine implementation}.
\newblock B.I. Wissenschaftsverlag, Mannheim-Leipzig-Wien-Z{\"u}rich, 1994.

\bibitem{SchonhageS1971}
A.~Sch{\"o}nhage and V.~Strassen.
\newblock Schnelle {M}ultiplikation grosser {Z}ahlen.
\newblock {\em Computing}, 7:281--292, 1971.

\bibitem{ShepherdsonS63}
J.~C. Shepherdson and H.~E. Sturgis.
\newblock Computability of recursive functions.
\newblock {\em Journal of the ACM}, 10(2):217--255, Apr. 1963.

\bibitem{Toom1963}
A.~L. Toom.
\newblock The complexity of a scheme of functional elements simulating the
  multiplication of integers.
\newblock {\em Dokl. Akad. Nauk SSSR}, 150:496--498, 1963.
\newblock (in Russian). English translation in Soviet Mathematics 3, 714-716,
  1963.

\bibitem{Wikipedia2013}
{Wikipedia, the free encyclopedia}.
\newblock Multiplication algorithm.
\newblock http://en.wikipedia.org/wiki/Multiplication\_algorithm, October 2013.

\end{thebibliography}
\end{document}